\documentclass[11pt]{article}

\usepackage[colorlinks=true, allcolors=blue]{hyperref}
\usepackage[english]{babel}

\usepackage[letterpaper,top=1in,bottom=1in,left=1in,right=1in]{geometry}

\usepackage[nolist,nohyperlinks]{acronym}

\usepackage{scrextend}
\usepackage[dvipsnames]{xcolor}
\newcommand{\todo}[1]{\noindent\textcolor{Green}{\textsc{todo:} #1}}

\usepackage{amsmath, amscd, amsthm, amsfonts}
\usepackage{amssymb}
\usepackage{mathtools}
\usepackage{thmtools}
\usepackage{hyperref}
\usepackage{subfig}
\usepackage{graphicx}
\usepackage{svg}   
\usepackage{enumerate}  
\usepackage{csquotes} 
\usepackage{booktabs}  
\usepackage{multirow}
\usepackage{braket} 
\usepackage{verbatim}

\usepackage{dsfont}

\newtheorem{theorem}{Theorem}
\newtheorem{observation}{Observation}[section]
\newtheorem{corollary}[theorem]{Corollary}
\newtheorem{definition}{Definition}

\newtheorem{lemma}{Lemma}

\newcommand{\zz}{\mathbb{Z}}
\newcommand{\nn}{\mathbb{N}}
\newcommand{\qq}{\mathbb{Q}}

\DeclarePairedDelimiter{\nint}\lfloor\rceil
\DeclarePairedDelimiter{\floor}\lfloor\rfloor
\newcommand{\norm}[1]{\left\lVert #1 \right\rVert}
\newcommand{\abs}[1]{\left\lvert #1 \right\rvert}

\DeclareMathOperator{\lattice}{\mathcal{L}}

\DeclareMathOperator{\OTilde}{\ensuremath{\tilde{O}}}

\DeclareMathOperator{\rem}{rem}

\newcommand{\ie}{i.\,e.}

\newcommand{\eg}{e.\,g.}

\usepackage{algorithmicx}
\usepackage[noend]{algpseudocode}
\usepackage{algorithm}
\newcommand*\Let[2]{\State #1 $\gets$ #2}

\algrenewcommand\alglinenumber[1]{\sf\footnotesize \texttt{#1\,}}%

\algrenewcommand\algorithmicrequire{\textbf{Precondition:}}
\algrenewcommand\algorithmicensure{\textbf{Postcondition:}}

\usepackage{multicol}
\usepackage{tcolorbox}

\title{Simple Lattice Basis Computation - The Generalization\\ of the Euclidean Algorithm \thanks{This research was supported by German Research Foundation (DFG) project KL 3408/1-1}}
\author{
	Kim-Manuel Klein
	\\\small Kiel University
	\\\small kimmanuel.klein@uni-luebeck.de
	\and
	Janina Reuter
	\\\small Kiel University
	\\\small janina.reuter@email.uni-kiel.de
}
\date{}

\begin{document}
\maketitle
\begin{abstract}
The Euclidean algorithm is one of the oldest algorithms known to mankind. Given two integral numbers $a_1$ and $a_2$, it computes the greatest common divisor (gcd) of $a_1$ and $a_2$ in a very elegant way. From a lattice perspective, it computes a basis of the sum of two one-dimensional lattices $a_1 \mathbb{Z}$ and $a_2 \mathbb{Z}$ as $\gcd(a_1,a_2) \mathbb{Z} = a_1 \mathbb{Z} + a_2 \mathbb{Z}$. In this paper, we show that the classical Euclidean algorithm can be adapted in a very natural way to compute a basis of a general lattice $\lattice(a_1, \ldots , a_m)$ given vectors $a_1, \ldots , a_m \in \mathbb{Z}^n$ with $m> \mathrm{rank}(a_1, \ldots ,a_m)$. Similar to the Euclidean algorithm, our algorithm is very easy to describe and implement and can be written within 12 lines of pseudocode.

While the Euclidean algorithm halves the largest number in every iteration, our generalized algorithm halves the determinant of a full rank subsystem leading to at most $\log (\det B)$ many iterations, for some initial subsystem $B$. Therefore, we can compute a basis of the lattice using at most $\OTilde((m-n)n\log(\det B) + mn^{\omega-1}\log(\norm{A}_\infty))$ arithmetic operations, where $\omega$ is the matrix multiplication exponent and $A = (a_1, \ldots, a_m)$. Even using the worst case Hadamard bound for the determinant, our algorithm improves upon existing algorithm.

Another major advantage of our algorithm is that we can bound the entries of the resulting lattice basis by $O(n^2\cdot \norm{A}_{\infty})$ using a simple pivoting rule. This is in contrast to the typical approach for computing lattice basis, where the Hermite normal form (HNF) is used. In the HNF, entries can be as large as the determinant and hence can only be bounded by an exponential term.

\end{abstract}


\section{Introduction}\label{sec:intro}
Given two integral numbers $a_1$ and $a_2$, the Euclidean algorithm computes the greatest common divisor (gcd) of $a_1$ and $a_2$ in a very elegant way. Starting with $s = a_1$ and $t= a_2$, a residue $r$ is being computed by setting
\begin{align*}
    r = \min_{x \in \zz} \{ r \in \zz \mid s x + r = t \} = \min \{ t \pmod s, |(t \pmod s) - s| \}.
\end{align*}
This procedure is continued iteratively with $s = t$ and $t = r$ until $r$ equals $0$. Since $r \leq \lfloor t/2 \rfloor$ the algorithm terminates after at most $\log(\min \{ a_1, a_2\})$ many iterations.

An alternative interpretation of the gcd or the Euclidean algorithm is the following: Consider all integers that are divisible by $a_1$ or respectively $a_2$, which is the set $a_1 \zz$ or respectively the set $a_2 \zz$. Consider their sum (i.e. Minkowski sum)
\begin{align*}
    A = a_1 \zz + a_2 \zz = \{a + b \mid a \in a_1 \zz, b \in a_2 \zz \}.
\end{align*}
It is easy to see that the set $A$ can be generated by a single element, which is the gcd of $a_1$ and $a_2$, i.e.
\begin{align*}
    a_1 \zz + a_2 \zz = gcd(a_1, a_2) \zz.
\end{align*}
Furthermore, the set $\mathcal{L} = a_1 \zz + a_2 \zz$ is closed under addition, subtraction and scalar multiplication, which is why all values for $s,t$ and $r$, as defined above in the Euclidean algorithm, belong to $\mathcal{L}$. In the end, the smallest non-zero element for $r$ obtained by the algorithm generates $\mathcal{L}$ and hence $\mathcal{L} = r \zz = gcd(a_1,a_2) \zz$.

This interpretation does not only allow for an easy correctness proof of the Euclidean algorithm, it also allows for a generalization of the algorithm into higher dimensions. For this, we consider vectors $A_1, \ldots , A_{m} \in \zz^n$ with $m>n$ and the set of points in space generated by sums of integral multiples of the given vectors, i.e.
\begin{align*}
    \mathcal{L} = A_1 \zz + \ldots A_{m} \zz.
\end{align*}
The set $\mathcal{L}$ is called a \emph{lattice} and is generally defined for a given matrix $A$ with column vectors $A_1, \ldots , A_m$ by
\begin{align*}
    \mathcal{L}(A) = \{ \sum_{i=1}^{m} \lambda_i A_i \mid \lambda_i \in \zz \}.
\end{align*}
One of the most basic facts from lattice theory is that every lattice $\mathcal{L}$ has a basis $B$ such that $\mathcal{L}(B) = \mathcal{L}(A)$, where $B$ is a square matrix.
Note that the set $a_1 \zz + a_2 \zz$ is simply a one-dimensional lattice and in this sense the Euclidean algorithm simply computes a basis of the one-dimensional lattice with $gcd(a_1,a_2) \zz = a_1 \zz + a_2 \zz$.

Hence, morally, a multidimensional version of the Euclidean algorithm should compute for a given matrix $A = (A_1, \ldots , A_{m})$ a basis $B \in \zz^{n \times n}$ such that
\begin{align*}
    \mathcal{L}(B) = \mathcal{L}(A).
\end{align*}
The problem of computing a basis for the lattice $\lattice(A)$ is called \emph{lattice basis computation}. 
In this paper, we show that the classical Euclidean algorithm can be generalized in a very natural way to do just that. Using this approach, we improve upon the running time of existing algorithms for lattice basis computation.

\subsection{Lattice Basis computation}
\acused{hnf}\acused{snf}\acused{de}\acused{lbr}

The first property of a lattice that is typically taught in a lattice theory lecture is the fact that each lattice has a basis. Computing a basis of a lattice is one of the most basic algorithmic problems in lattice theory. Often it is required as a subroutine by other algorithms~\cite{DBLP:journals/jsc/Pohst87,DBLP:conf/eurocal/BuchmannP87,DBLP:conf/stoc/GentryPV08, DBLP:books/daglib/0018102}. There are mainly two methods on how a basis of a lattice can be computed. The most common approaches rely on either a variant of the LLL-algorithm or on computing the Hermite normal form (HNF), where the fastest algorithms all rely on the HNF. Considering these approaches however, one encounters two major problems. First, the entries of the computed basis can be as large as the determinant and therefore exponential in the dimension. Secondly and even worse, intermediate numbers on the computation might even be exponential in their bit representation. This effect is called intermediate coefficient swell. Due to this problem, it is actually not easy to show that a lattice basis can be computed in polynomial time. Kannan und Buchem~\cite{DBLP:journals/siamcomp/KannanB79} were the first ones to show that the intermediate coefficient swell can be avoided when computing the HNF and hence a lattice basis can actually be computed in polynomial time. The running time of their algorithm was later improved by  Chou and Collins\cite{DBLP:journals/siamcomp/ChouC82} and Iliopoulos \cite{DBLP:journals/siamcomp/Iliopoulos89}.


Recent and the most efficient algorithms for lattice basis computation all rely on computing the HNF, with the most efficient one being the algorithm by Storjohann and Labahn~\cite{DBLP:conf/issac/StorjohannL96}. Given a full rank matrix $A\in\zz^{n\times m}$ the \acs{hnf} can be computed by using only $\OTilde(n^{\omega}m\cdot \log\norm{A}_\infty)$ many bit operations. 
The algorithm by Labahn and Storjohann~\cite{DBLP:conf/issac/StorjohannL96} improves upon a long series of papers~\cite{DBLP:journals/siamcomp/KannanB79, DBLP:journals/siamcomp/ChouC82,DBLP:journals/siamcomp/Iliopoulos89} and has not been improved since its publication in 1996. Only in the special case that $m-n = 1$, Li and Storjohann~\cite{DBLP:conf/issac/LiS22} manage to obtain a better running time that essentially matches matrix multiplication time. 

Other recent paper considering lattice basis computation focus on properties other than improving the running time. There are several algorithms that preserve orthogonality from the original matrix, \eg{} $\norm{B^*}_\infty \leq \norm{A^*}_\infty$, or improve on the $\ell_\infty$ norm of the resulting matrix~\cite{DBLP:conf/stoc/NovocinSV11,DBLP:conf/issac/NeumaierS16}, or both~\cite{DBLP:conf/crypto/HanrotPS11,DBLP:conf/issac/LiN19,DBLP:conf/focs/CaiN97,DBLP:books/daglib/0018102}. 
Except for the \acs{hnf} based basis algorithm by Lin and Nguyen~\cite{DBLP:conf/issac/LiN19}, all of the above algorithms have a significantly higher time complexity compared to Labahn's and Storjohann's \acs{hnf} algorithm. 
The algorithm by Lin and Nguyen use existing HNF algorithms and apply a separate coefficient reduction algorithm resulting in a basis with $\ell_\infty$ norm bounded by $n\norm{A}_\infty$.

\subsection{Our Contribution}
In this paper we develop a fundamentally new approach for lattice basis computation given a matrix $A$ with column vectors $A_1, \ldots, A_m \in \zz^n$. Our approach does not rely on any normal form of a matrix or the LLL algorithm. Instead, we show a direct way to generalize the classical Euclidean algorithm to higher dimensions. After a thorough literature investigation and talking to many experts in the area, we were surprised to find out that this approach actually seems to be new.

Our approach does not suffer from intermediate coefficient growth and hence gives an easy way to show that a lattice basis can be computed in polynomial time. Furthermore, we can show that by an easy pivoting rule the resulting lattice basis has only a mild coefficient growth compared to the absolute values of the entries in the $A_i$ vectors. We can show that the entries of the resulting basis can be bounded by $O(n^2 \cdot \norm{A}_\infty)$.

Similar to the Euclidean algorithm, our algorithm chooses an initial basis $B$ from the given vectors and updates the basis according to a remainder operation and then exchanges a vector by this remainder. In every iteration, the determinant of $B$ decreases by a factor of at least $1/2$ and hence the algorithm terminates after at most $\log \det(B)$ many iterations. Similar to the Euclidean algorithm, our algorithms can be easily described and implemented. 

We develop data structures for our novel algorithmic approach and analyze the running time of our algorithms comparing to state of the art algorithms for lattice basis computation.
But first, how do we measure efficiency in the running time of algorithms for lattice basis computation? There are mainly two different ways on how this can be done. First, one can simply count the number of arithmetic operation that the algorithm performs. In this model, one does not care about the size of the numbers and simply counts each basic ring operation: addition, subtraction, multiplication, and division. This concept of arithmetic complexity is often used in the context of matrix related problems (e.g.~\cite{DBLP:journals/iandc/Schnorr06, DBLP:conf/focs/BrandNS19, DBLP:conf/soda/ChepurkoCKW22}) and linear programming~(e.g.~\cite{DBLP:conf/soda/Brand20}), for example the concept of strong polynomiality relies on the notion of arithmetic complexity. 

A more precise measure of the running time of an algorithm uses the so called \emph{bit complexity} model. Here, one counts each bit operation and hence for example an addition of two numbers of size $t$ bits requires $O(t)$ bit operations.

In most algorithmic problems the arithmetic model and the bit complexity model do not need to be distinguished as the respective running times would essentially match.
However, this is not the case for lattice basis computation (and related problems). For example, intermediate numbers in computing the Hermite normal form can become exponentially large in the dimension compared to the input numbers. Therefore, the same algorithm might have an additional factor in the bit complexity model compared to the arithmetic complexity.


\subsubsection*{Arithmetic Complexity}
While the bit complexity model is more precise in terms of worst case complexity, we also study our algorithms within the notion of arithmetic complexity. The main advantage of this model is that it provides a relatively easy analysis of the running time. Also, as one is simply counting the number of elementary ring operations the model provides an easier understanding of the running time when generalizing to other algebraic structure. Historically however, it was often the case that in the end, the same running time in the bit complexity model could be achieved as in the model of arithmetic complexity. But for the bit complexity to achieve the same running time, typically a very thorough analysis on the bit level is necessary. Consider for example the classical Euclidean algorithm when applied to numbers of bit length $t$. The algorithm requires $O(t^2)$ many bit operations, while only $O(t)$ many arithmetic operations are necessary. Using rather sophisticated operations on the bit level however, Schönhage~\cite{DBLP:journals/acta/Schonhage71} developed an algorithm computing the gcd by using only $\OTilde(t)$ many bit operations. 

In terms of arithmetic complexity, our main result is to develop an algorithm which uses at most 
\[\OTilde(\log\det(B) \cdot (m-n)n + mn^{\omega -1}\log ||A||_\infty).\]
many arithmetic operations. Even with a worst case Hadamard bound for $\det(B) \leq (n ||A||_\infty)^n$ and bounding $(m-n) \leq m$, we obtain a running time of $\OTilde(mn^2 \log ||A||_\infty)$ and hence improve upon the algorithm of Storjohann and Labahn~\cite{DBLP:conf/issac/StorjohannL96} by a factor of $n^{\omega -2} \approx n^{0.37}$ for current values of $\omega$. We are not aware of any other algorithms with a better running time within the arithmetic complexity model. But note that the algorithm by Storjohann and Labahn has the same time complexity within the bit complexity model, while our algorithms perform slightly worse within the bit complexity model. However, we are confident that a sophisticated analysis on the bit level similar to the approach of Schönhage~\cite{DBLP:journals/acta/Schonhage71}, will provide a much better running time also in the bit complexity model. In this sense, we see our results within the arithmetic complexity model as the potential that the presented approach has. Recall that our approach is new and builds upon very few subroutines while competing with algorithms for the HNF which build upon decades of research across dozens of papers.


\subsubsection*{Bit Complexity}
When it comes to the bit complexity model, in general, one has to pay attention to the growth of intermediate numbers in the matrix and in the respective solutions of linear systems. In the case of computing the HNF, this problem is typically dealt with by applying a separate coefficient reduction algorithm. In case of our algorithm however, we can completely ignore this issue. We show that for an easy pivoting rule, we only have quadratically growing coefficients in our basis matrix $B$. As a result, we can improve upon the running time of the algorithm by Labahn and Storjohann~\cite{DBLP:conf/issac/StorjohannL96} in the case that $m-n$ is small. Our algorithm requires $\OTilde((m-n)n^3 \log^2 ||A||_\infty)$ bit operations and therefore yields an improved running time if the number of vectors that need to be merged into the basis is small, i.e. $m-n \in O(n^{\omega-2})$. In the case that $\det(B)$ is small, we also obtain an improved running time. For the general case, our algorithm matches the running time of \cite{DBLP:conf/issac/StorjohannL96} in terms of $m$ and $n$ having a bit complexity of 
\begin{align*}
    \OTilde(mn^\omega\log^2||A||_\infty).
\end{align*}
We are rather confident that the quadratic term in $\log||A||_\infty$ can be improved to a single logarithmic term by using an approach similar to Schönhage~\cite{DBLP:journals/acta/Schonhage71}. However, the required observations on the bit level would exceed the scope of this paper.

Furthermore, our algorithms can be easily modified to compute the determinant of a square matrix $B$ or compute a solution for a linear system of Diophantine equations. In the case of computing the determinant, the running times of the algorithms remain the same. However, in the case of computing a solution for Diophantine systems, the worst case complexity of the algorithms increase.





\section{Algorithm Sketch}\label{sec:algorithm_sketch}
In this section and throughout the paper, we assume that $rank(A) = n$ and therefore the lattice $\lattice(A)$ is full dimensional. However, our algorithms can be applied in a similar way if $rank(A) < n$. 
The term $(m-n)$ in the running times of the respective algorithm (which represents the number of vectors that need to be merged into the basis) is then replaced by the term $(m-rank(A))$.

\subsection*{Preliminaries}
Consider a lattice $\lattice(B)$ for a given full dimensional basis $B \in \zz^{n \times n}$. An important notion that we need is the so called \emph{fundamental parallelepiped}
\begin{align*}
    \Pi(B) = \{ B x \mid x \in [0,1)^n \}
\end{align*}
\begin{figure}
    \centering
    \includegraphics[width=0.2\textwidth]{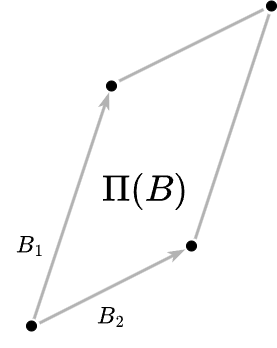}
    \caption{The parallelepiped of $B=(B_1B_2)$.}
    \label{fig:parallelepiped}
\end{figure}
see also \autoref{fig:parallelepiped}. As each point $a \in \mathbb{R}^n$ can be written as
\begin{align*}
    a = B \floor{x} + B \{ x \},
\end{align*}
it is easy to see that the space  $\mathbb{R}^n$ can be partitioned into parallelepipeds. Here, $\floor{x}$ denotes the vector, where each component $x_i$ is rounded down and $\{ x \} = x - \floor{x}$ is the vector with the respective fractional entries $x_i \in [0,1)$.
In fact, the notion of $\Pi(B)$ allows us to define a multi-dimensional modulo operation by mapping any point $a \in \zz^n$ to the respective \emph{residue vector} in the parallelepiped $\Pi(B)$, i.e. 
\begin{align*}
    a \pmod {\Pi(B)} := B \{ B^{-1} a \} \in \Pi(B).
\end{align*}
Furthermore, for $a \in \zz$, we denote with $\nint{z}$ the next integer from $a$, which is $\floor{a + 1/2}$. When we use these notations on a vector $a \in \zz^n$, the operation is performed entry-wise. 

Note that the parallelepiped $\Pi(B)$ has the nice property, that its volume as well as the number of contained integer points is exactly $\det(B)$, i.e. 
\begin{align*}
    vol(\Pi(B)) = |\Pi(B) \cap \zz^n| = \det(B).
\end{align*}

In our algorithm, we will change our basis over time by exchanging column vectors. We denote the exchange of column $i$ of a matrix $B$ with a vector $v$ by $B\setminus B_i \cup v$. The notation $B\cup v$ for a matrix $B$ and a vector $v$ of suitable dimension denotes the matrix, where $v$ is added as another column to matrix $B$. Similarly, the notation $B \cup S$ for a matrix $B$ and a set of vectors $S$ (with suitable dimension) adds the vectors of $S$ as new columns to matrix $B$. While the order of added columns is ambiguous, we will use this operation only in cases where the order of column vectors does not matter.

\subsection*{The Algorithm}
Given two numbers, the classical Euclidean algorithm, essentially consists of two operations. First, a \emph{modulo operation} computes the modulo of the larger number and the smaller number. Second, an \emph{exchange operation} discards the larger number and adds the remainder instead. The algorithm continues with the smaller number and the remainder. 

Given vectors $A=\{A_1, \ldots , A_{n+1}\} \subset \zz^n$, our generalized algorithm performs a multi-dimensional version of \emph{modulo} and \emph{exchange operations} of columns with the objective to compute a basis $B \in \zz^{n \times n}$ with $\lattice (B) = \lattice (A)$. First, we choose $n$ linearly independent vectors from $A$ which form a non-singular matrix $B$. The lattice  $\lattice(B)$ is a sub-lattice of $\lattice(A)$. Having this sub-basis, we can perform a division with residue in the lattice $\lattice (B)$. Hence, the remaining vector $a \in A \setminus B$ can be represented as
\begin{align*}
    a = B \floor{B^{-1}a} +r,
\end{align*}
where $r$ is the remainder $a \pmod{\Pi (B)}$, see also \autoref{fig:high_dimension_modulo}. 
\begin{figure}[t]
   \centering
      \subfloat[The modulo operation in dimension $2$.
      \label{fig:high_dimension_modulo}]{{\small
    \includegraphics[width=0.40\textwidth]{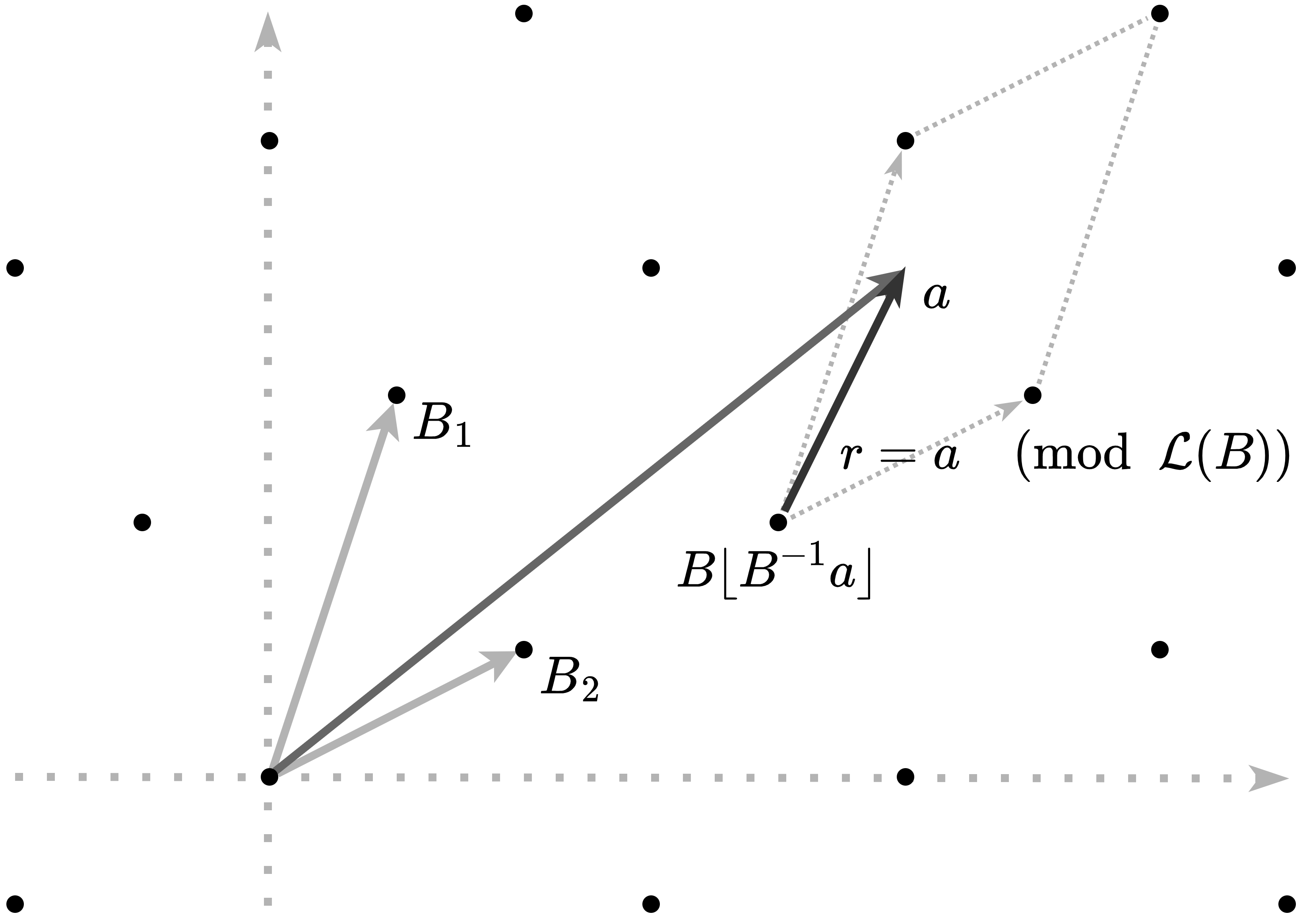}}}\qquad\qquad\qquad
      \subfloat[Exchange of a basis vector and the parallelopipeds for $B_1$ and $B_2$ (solid), $B_2$ and $r$ (dotted), and $B_2$ and $r'$ (dashed).
      \label{fig:basis_exchange}]{{\small
    \includegraphics[width=0.40\textwidth]{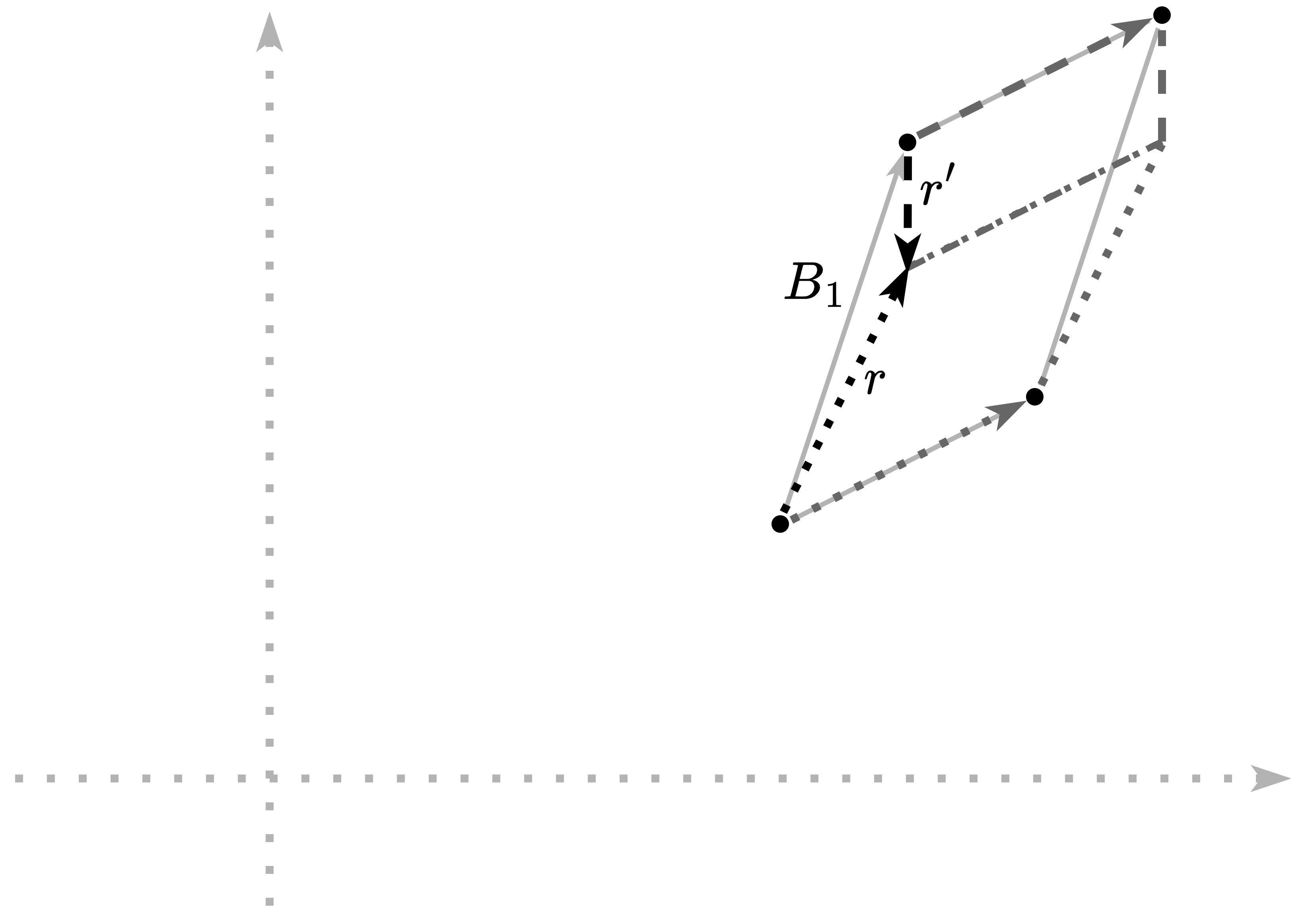}}}\\
    \caption{The modulo operation with respect to a lattice and the exchange operation depending on $\nint{x_1}$.}
\end{figure}
In dimension $n= 1$ this is just the classical division with residue and the corresponding modulo operation, \ie{} $a = b \cdot \floor{a/b} + r$.

Having the residue vector $r$ at hand, the \emph{exchange step} of our generalized version of the Euclidean algorithm exchanges a column vector of $B$ with the residue vector $r$. In dimension $>1$, we have the choice on which column vector to discard from $B$. The choice we make is based on the solution $x \in \qq^n$ of the linear system $Bx = a$.
\begin{itemize}
    \item Case 1: $x \in \zz^n$. In the case that the solution $x$ is integral, we know that $a \in \lattice (B)$ and hence $\lattice (B \cup a ) = \lattice (B)$. Our algorithm terminates.
    \item Case 2: There is a fractional component $i$ of $x$. In this case, our algorithm exchanges $B_i$ with $r$, \ie{} $B' = B \setminus  B_i  \cup  r $.
\end{itemize}
The algorithm iterates this procedure with basis $B'$ and vector $a = B_i$ until Case 1 is achieved.\medskip

\begin{addmargin}[3.5em]{3.5em}
\begin{tcolorbox}[
sharp corners=all,
colback=white,
colframe=black,
size=tight,
boxrule=0.2mm,
left=3mm,right=3mm,top=3mm,bottom=3mm
]
{\begin{multicols}{2}

\textbf{Euclidean Algorithm}\bigskip

\textsc{Modulo Operation}\\
$t = s\lfloor s^{-1} t\rfloor + r$ \medskip

\textsc{Exchange Operation}\\
$t = s, \ s=r$ \medskip

\textsc{Stop Condition}\\
$s^{-1} t\ $ is integral

\columnbreak

\textbf{Generalized Euclidean Algorithm}\bigskip

\textsc{Modulo Operation}\\
$a = B\floor{B^{-1}a} + r$ \medskip

\textsc{Exchange Operation}\\
$a = B_i, \ B_i := r$ \medskip

\textsc{Stop Condition}\\
$B^{-1}a\ $ is integral

\end{multicols}}
\end{tcolorbox}
\end{addmargin}

Two questions arise: Why is this algorithm correct and why does it terminate?

\textbf{Termination:}\\
The progress in step 2 can be measured in terms of the determinant. For $x$ with $Bx=a$ the exchange step in case 2 swaps $B_i$ with $r = B \{ x \}$ and $\{x_i\} \neq 0$ to obtain the new basis $B'$. By Cramer's rule we have that $\{x_i\} = \frac{\det B'}{\det B}$ and hence the determinant decreases by a factor of $\{x_i\} < 1$. The algorithm eventually terminates since $\det(\lattice(A))\geq 1$ and all involved determinants are integral since the corresponding matrices are integral. A trivial upper bound for the number of iterations is  the determinant of the initial basis. 

\textbf{Correctness:}\\
Correctness of the algorithm follows by the observation that $\lattice(B \cup a) = \lattice(B \cup r)$. To see this, it is sufficient to prove $a \in \lattice(B \cup r)$ and $r \in \lattice(B \cup a)$. By the definition of $r$ we get that $a = Bx = B\floor{x} + B\{x\} = B\floor{x} + r$. Hence, $a$ and $r$ are integral combinations of vectors from $B\cup r$ and $B\cup a$, respectively, and hence $\lattice(B\cup a) = \lattice(B \cup r)$.


The multiplicative improvement of the determinant in step 2 can be very close to $1$, \ie{} $\frac{\det(B)-1}{\det(B)}$. In the classical Euclidean algorithm a step considers the remainder $r$ for $a = b\floor{a/b} + r$. The variant described in \autoref{sec:intro} considers an $r'$ for $a = b\nint{a/b} + r'$. Taking the next integer instead of rounding down ensures that in every step the remainder in absolute value is at most half of the size of $b$. Our generalized Euclidean algorithm uses a modified modulo operation that does just that in a higher dimension. In our case, this modification ensures that the absolute value of the determinant decreases by a multiplicative factor of at most $1/2$ in every step as we explain below. The number of steps is thus bounded by $\log\det(B)$. The generalization to higher dimensions chooses $i$ such that $x_i$ is fractional and rounds it to the next integer $\nint{x_i}$ while the other entries of $x$ are again rounded $\floor{x_j}$ for $j\neq i$. Formally, this modulo variant is defined as 
\begin{align*}
    a \ \ (\bmod'\ \Pi(B)) := r' := a - (\sum_{j\neq i}B_j\floor{x_j} + B_i\nint{x_i})
\end{align*}
for $Bx=a$ and some $i$ such that $\{x_i\}\neq 0$. By Cramer's rule we get that the determinant decreases by a multiplicative value of at least $1/2$ in every iteration since $\frac{1}{2} \leq \abs{x_i - \nint{x_i}} = \abs{\frac{\det B'}{\det B}}$. In \autoref{fig:basis_exchange} the resulting basis for exchanging $B_1$ with $r = a \ \ (\bmod\  \Pi(B))$ and with $r'= a \ \ (\bmod'\  \Pi(B))$ shows that in both cases the volume of the parallelepiped decreases, which is equal to the determinant of the lattice. In \autoref{fig:application}, an example of our algorithm is shown.
\begin{figure}[t]
   \centering
      \subfloat[Application of our algorithm, $r'$ is the first remainder.
      \label{fig:alg_0}]{{\small
    \includegraphics[width=0.30\textwidth]{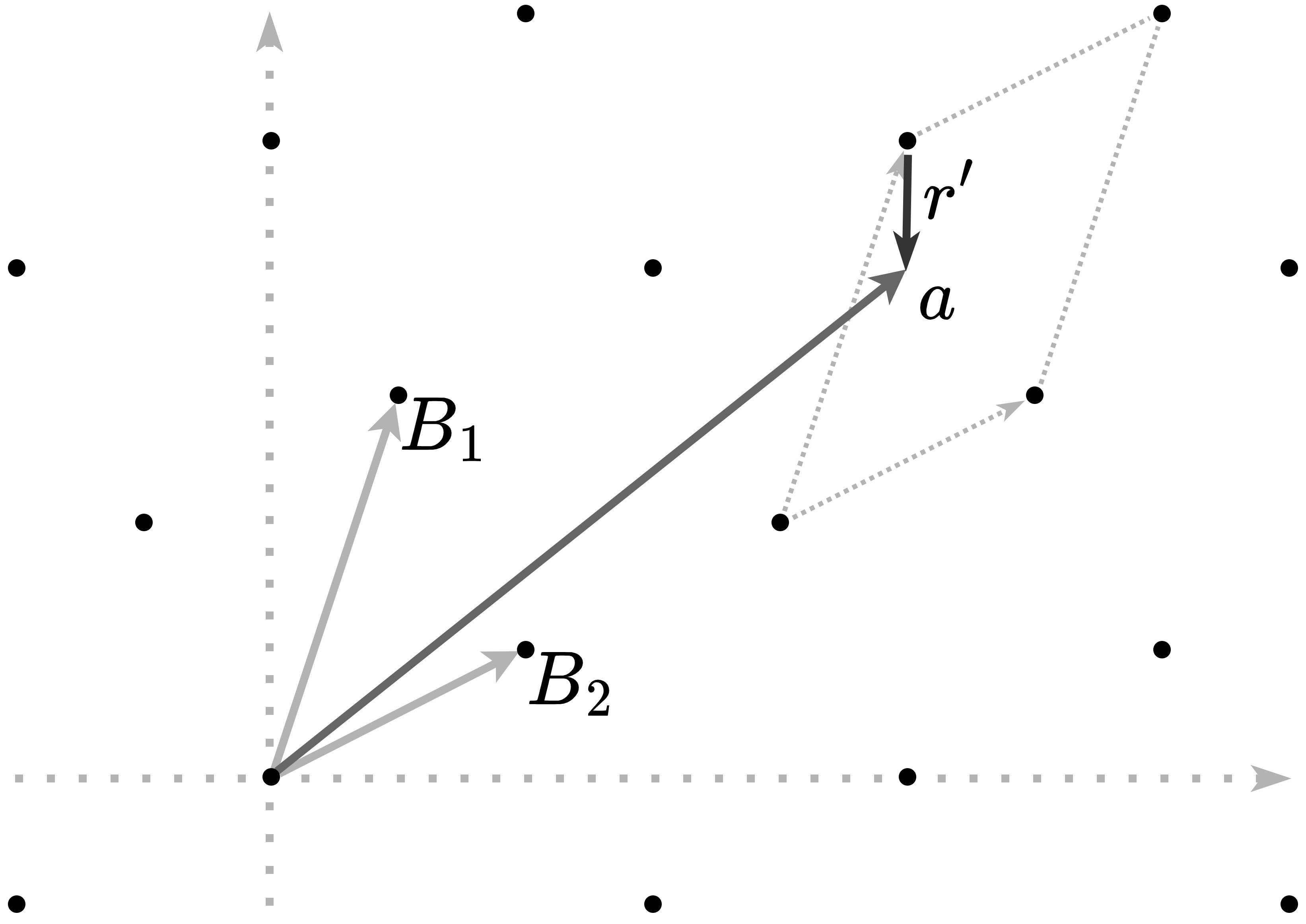}}} 
    \qquad
      \subfloat[Vectors $r'$ and $B_1$ were exchanged and $r''$ denotes the second remainder.
      \label{fig:alg_1}]{{\small
    \includegraphics[width=0.30\textwidth]{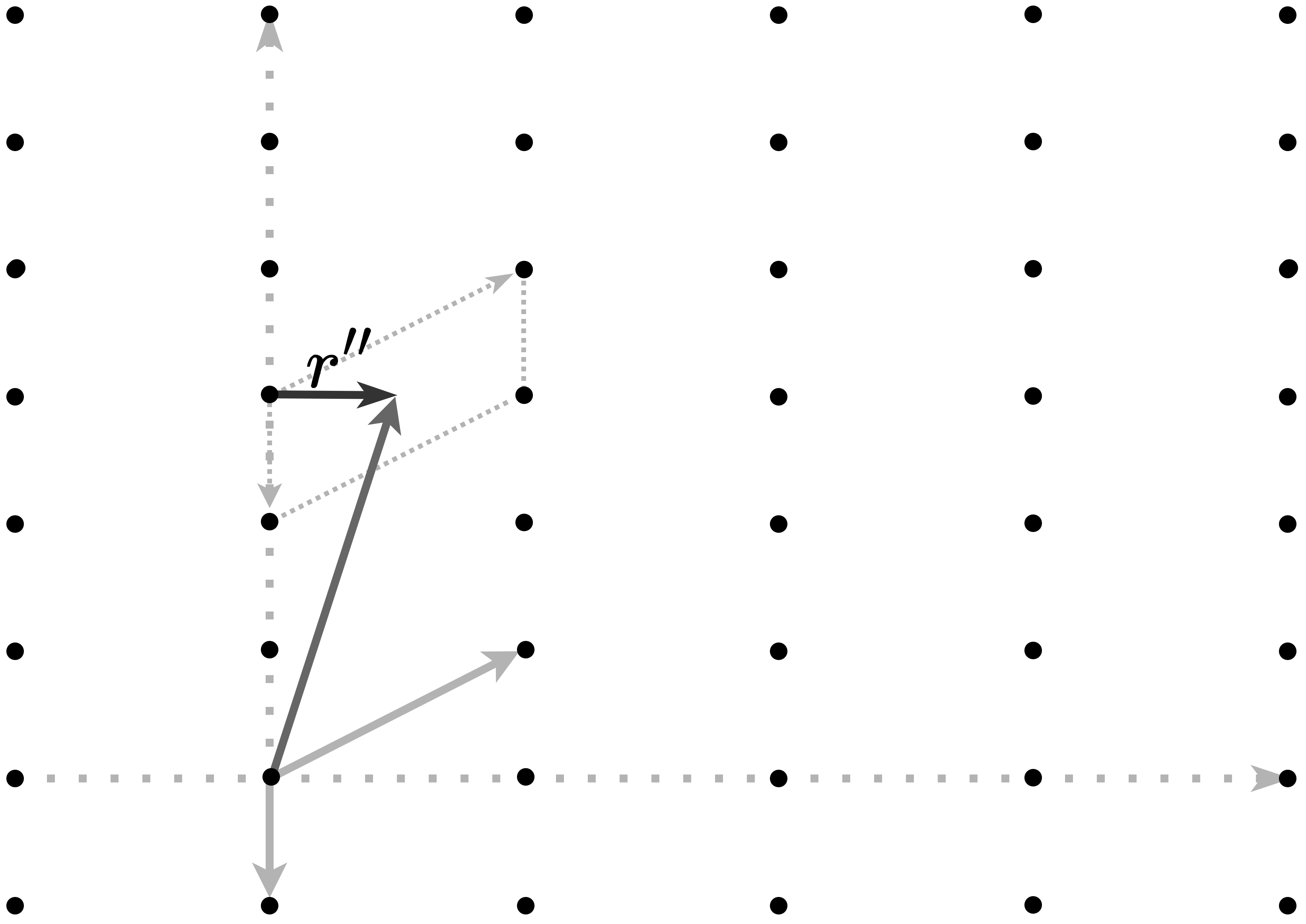}}}\qquad
      \subfloat[Vectors $r'$ and $B_2$ were exchanged. $B_2$ is in the lattice and the algorithm terminates.
      \label{fig:alg_2}]{{\small
    \includegraphics[width=0.30\textwidth]{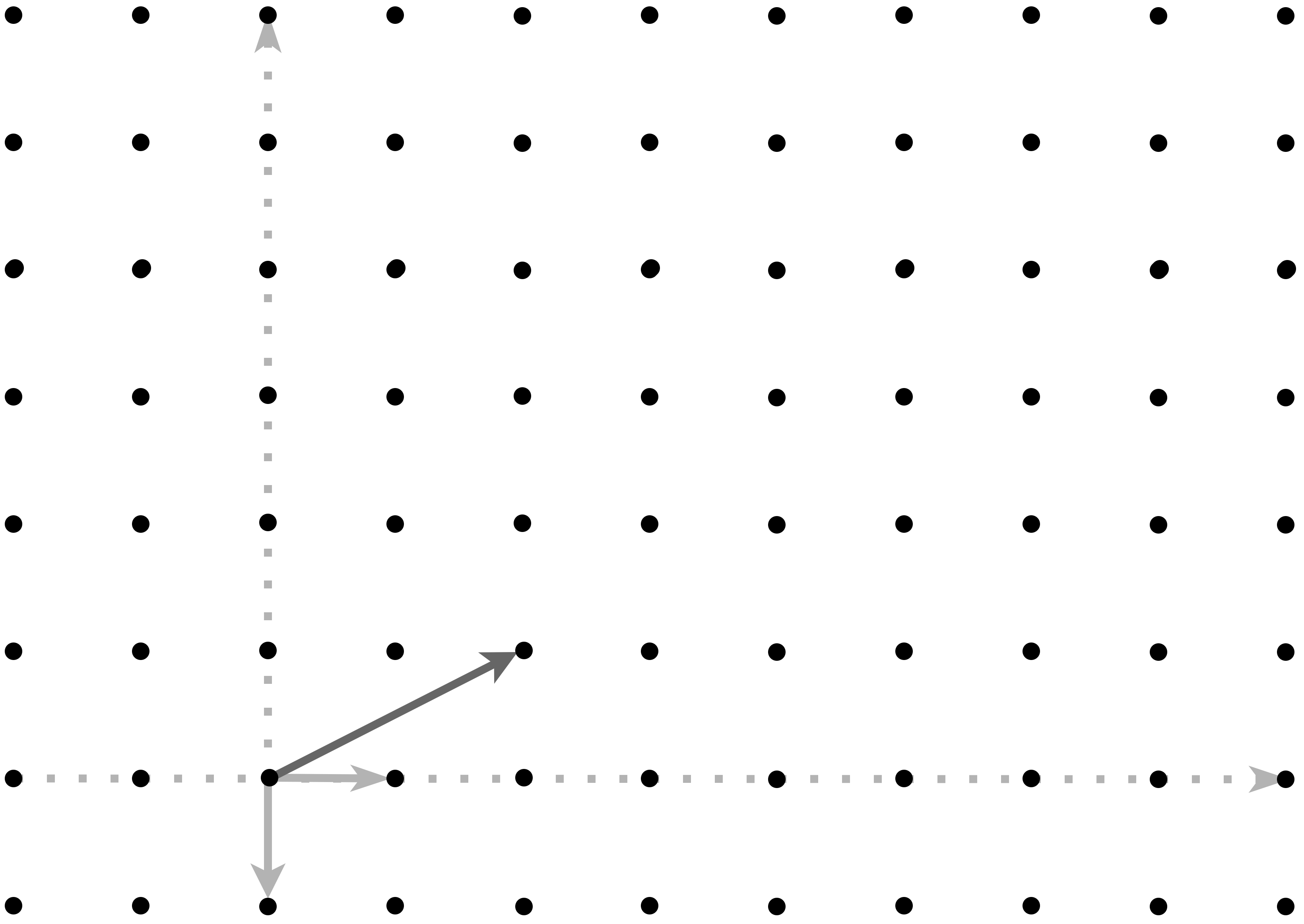}}}
   \caption{An application of the algorithm.}
    \label{fig:application}
\end{figure}


\subsection{Basic Algorithm}
In the following we state the previously described algorithm formally.
\begin{algorithm}
  \caption{Generalized Euclidean Algorithm (Basic Algorithm)
    \label{alg:easyversion}}
  \begin{algorithmic}[1]
    \Statex
    \textsc{Input: } A matrix $A=\left(A_1,\ldots, A_m \right) \in \zz^{n\times m}$ \\
    \textbf{find} independent vectors $B := \left(B_1, \ldots, B_n\right)$ with $B_i\in \{A_1,\ldots, A_m\}$
    \Let{$C$}{$\{A_1,\ldots, A_m\}\setminus \{B_1,\ldots, B_n\}$}
    \While{$C\neq \emptyset$}
      \State{\textbf{choose any} $c\in C$}
      \State{\textbf{solve} $Bx=c$ }
      \If{$x$ is integral}
      \Let{$C$}{$C\setminus\{c\}$}
      \Else
      \Let{$i$}{$\arg\min_{j\leq n, \,\{x_j\}\neq 0}\abs{x_j - \nint{x_j}}$}\Comment{Any $j$ s.t. $\{x_j\}\neq 0$ suffices.}
      \Let{$C$}{$C\setminus\{c\}\cup \{B_i\}$}
      \Let{$B_i$}{$c - (\sum_{j\neq i}B_j\floor{x_j}  + B_i \nint{x_i})$}
      
      \EndIf
    \EndWhile
    \State \Return{$B$}
  \end{algorithmic}
\end{algorithm}

\begin{theorem}\label{theorem:basic_version_correctness}
\autoref{alg:easyversion} computes a basis for the lattice $\mathcal{L}(A)$.
\end{theorem} 
\begin{proof}
Let us consider the following invariant.\medskip

\textit{Claim.} In every iteration $\mathcal{L}(A) = \mathcal{L}(B\cup C)$.

By the definition of $B$ and $C$ the claim holds in line 2. We need to prove that removing $c$ from $C$ in line $7$ and altering $B$ and $C$ in lines 9-11 do not change the generated lattice. In line 7 we found $c$ is an integral combination of vectors in $B$. Thus, every lattice point can be represented without the use of $c$ and $c$ can be removed without altering the generated lattice. In lines 9 and 10 a vector $c$ is removed from $B\cup C$ and instead a vector $c' = c - (\sum_{j\neq i}B_j\floor{x_j}  + B_i \nint{x_i})$ is added. By the definition of $c'$, the removed vector $c$ is an integral combination of vectors $c', B_1, \ldots, B_n$ and $c'$ is an integral combination of vectors $c, B_1, \ldots, B_n$. Using the same argument as above, this does not change the generated lattice. \medskip

The algorithm terminates when $C=\emptyset$. In this case $B$ is a basis of $\mathcal{L}(A)$, since by the invariant we have that $\mathcal{L}(B) = \mathcal{L}(B \cup C) = \mathcal{L}(A)$.

\end{proof}

\begin{observation}
\autoref{alg:easyversion} terminates after at most $\log\det(B^{(1)})$ exchange steps.
\end{observation}

\subsection{Arbitrary Rank of Lattice}
In the case that the lattice $\mathcal{L}(A)$ is not fully dimensional Algorithm \ref{alg:easyversion} can easily be modified to also function in that case. This can be done by using Lemma~\ref{lem:maximal_subsystem} to choose a maximum set of linear independent vectors from $A$ as our initial basis $B$. The algorithm then proceeds to work with a basis $B$ containing $rank(B)$ many vectors. Note that every other vector in $C$ is then still contained in the linear subspace of $B$ and hence the linear system of equalities in step 5 of the algorithm is always solvable.

The same argument can be applied to any of the presented algorithms in this paper. For simplicity we therefore omit this case and assume from now on that $\mathcal{L}(A)$ is fully dimensional. As mentioned, the term $(m-n)$ in the running times of the respective algorithm (which represents the number of vectors that need to be merged into the basis) must be replaced by the term $(m-rank(A))$.

\section{Arithmetic Operations}

The main bottleneck in terms of running time of Algorithm~\ref{alg:easyversion} is that in each iteration, the linear system $Bx = c$ (line 5) needs to be solved. In this section, we present two efficient algorithms for lattice basis computation that do this step more efficiently. Algorithm~\ref{alg:inverseversion} uses the inverse matrix to obtain the respective solutions. As the basis $B$ changes, the inverse matrix is being updated. In Algorithm~\ref{alg:solutionupdateversion}, we use an efficient data structure that manages the solutions for all vectors that are not in the basis. The data structure is built in a way that it can be updated efficiently when the basis changes.

We analyze the algorithms with respect to their arithmetic complexity. A subproblem that arises is to find a maximal set of linearly independent vectors. In our algorithms we use the following Lemma for this subproblem.
\begin{lemma}[\cite{DBLP:conf/issac/LiS22}]\label{lem:maximal_subsystem}
Let $A\in\zz^{\Tilde{m}\times \Tilde{n}}$ have full column rank. There exists an algorithm that finds indices $i_1, \ldots, i_{\Tilde{n}}$ such that $A_{i_1}, \ldots, A_{i_{\Tilde{n}}}$ are linearly independent using $\OTilde(\Tilde{m}\Tilde{n}^{\omega-1}\log\norm{A}_\infty)$ bit operations. 
\end{lemma}

\subsection{Via Matrix Inverse Updates}
This first algorithm uses the fact that updating the inverse of a matrix and computing a matrix-vector multiplication both only requires $O(n^2)$ arithmetic operations. Thereby, we need to compute the inverse only once using $\OTilde(n^\omega)$ arithmetic operations and in every iteration we only require a quadratic number of operations for solving the linear system and updating the inverse.

\begin{algorithm}
  \caption{Generalized Euclidean Algorithm (Matrix Inverse Updates)
    \label{alg:inverseversion}}
  \begin{algorithmic}[1]
    \Statex
    \textsc{Input: } A matrix $A=\left(A_1,\ldots, A_m \right) \in \zz^{n\times m}$ \\
    \textbf{find} independent vectors $B := \left(B_1, \ldots, B_n\right)$ with $B_i\in \{A_1,\ldots, A_m\}$
    \State{\textbf{compute} the inverse $B^{-1}$}
    \Let{$C$}{$\{A_1,\ldots, A_m\}\setminus \{B_1,\ldots, B_n\}$}
    \While{$\det(\mathcal L(B))\neq 1$ \textbf{and} $C\neq \emptyset$}
      \State{\textbf{choose any} $c\in C$}
      \State{\textbf{compute} $x \leftarrow B^{-1}c$ }
      \If{$x$ integral}
      \Let{$C$}{$C\setminus\{c\}$}
      \Else
      \Let{$i$}{$\arg\min_{j\leq n, \,\{x_i\}\neq 0}\abs{x_i - \nint{x_i}}$}\Comment{Any $i$ s.t. $\{x_i\}\neq 0$ suffices.}
      \Let{$C$}{$C\setminus\{c\}\cup \{B_i\}$}
      \Let{$B_i$}{$c - (\sum_{j\neq i}B_j\floor{x_j}  + B_i \nint{x_i})$}
      \State{\textbf{update} inverse $B^{-1}$ according to new column $B_i$}
      \EndIf
    \EndWhile
    \State \Return{$B$}
  \end{algorithmic}
\end{algorithm}
\begin{theorem}
\autoref{alg:inverseversion} computes a basis for the lattice $\mathcal{L}(A)$ using \[\OTilde((m-n)n^2 + mn^{\omega-1}\log\norm{A}_{\infty}+ n^3\log\norm{A}_{\infty})\] arithmetic operations.
\end{theorem} 
\begin{proof}
Correctness of the algorithm follows similar to \autoref{theorem:basic_version_correctness}. Using \autoref{lem:maximal_subsystem} the set of linearly independent columns can be found in $\OTilde(mn^{\omega-1}\log\norm{A}_{\infty})$ bit operations. The inverse can be computed in $\OTilde(n^\omega)$. In every iteration either a vector from $C$ is discarded or an exchange operation is performed.  Thus, the number of iterations can be bounded by $m-n + \log(\det B^{(1)}) \leq m-n + n\log(n\norm{A}_{\infty})$, where $B^{(1)}$ is the matrix of linearly independent columns found in line 1 and the inequality follows the worst-case Hadamard bound on determinants. In every iteration a constant number of vector operations and matrix-vector multiplications is computed. Moreover, the inverse can be updated in $O(n^2)$ arithmetic operations, see e.g. Sherman and Morrison~\cite{sherman1949adjustment, sherman1950adjustment}. Therefore, the number of arithmetic operations used is bounded by $\OTilde((m-n)n^2 + mn^{\omega-1}\log\norm{A}_{\infty}+ n^3\log\norm{A}_{\infty}) \leq \OTilde(mn^2\log\norm{A}_{\infty})$.
\end{proof}

\subsection{Via System Solving}
The running time of the following algorithm improves on the previous one in the case that either $m-n$ or $\log(\det B^{(1)})$ is small. Instead of updating the inverse matrix in order to solve the next linear system, \autoref{alg:solutionupdateversion} computes all solutions at once and then updates the solution matrix. 

\begin{lemma}\label{lem:updated_solution_matrix}
Consider two matrices $B\in\zz^{\Tilde{n}\times \Tilde{n}}$ and $C\in\zz^{\Tilde{n}\times \Tilde{m}}$, where $B$ is full rank. Let $X:=B^{-1}C$ and consider an exchange step
\begin{align*}
    C' = C \setminus \{C_j\} \cup B_i, \qquad B'_i := C_j - \sum_{k\neq i}B_k\floor{X_{kj}} + B_i \nint{X_{ij}},
\end{align*}
where the $i$th column of $B$ is updated according to right-hand side $C_j$.
Then the updated solution matrix $X' := (B')^{-1} C'$ can be computed by
\begin{align*}
    X'_{ij} &= \frac{1}{X_{ij} - \nint{X_{ij}}}\\
    X'_{i\ell} &=  \frac{X_{i\ell}}{X_{ij} - \nint{X_{ij}}} \quad \text{for all } \ell\neq j\\
    X'_{kj} &=  \frac{-\{X_{kj}\}}{X_{ij} - \nint{X_{ij}}} \quad \text{for all } k\neq i\\
    X'_{k\ell} &= X_{k\ell} -  \frac{X_{i\ell} \cdot \{X_{kj}\}}{X_{ij} - \nint{X_{ij}}}  \quad \text{for all } \ell\neq j \text{ and } k\neq i.
\end{align*}
\end{lemma}
\begin{proof}
Since $C_j = BX_{*j}$, we can reformulate the exchange step as $B'_i := \sum_{k\neq i}B_k\{X_{kj}\} + B_i (X_{ij} - \nint{X_{ij}})$. As $B'_k = B_k$ is unchanged for $k\neq i$ we get that
\begin{align}\label{eq:old_Bi_in_terms_of_new_basis}
    B_i = B'_i\frac{1}{X_{ij} - \nint{X_{ij}}} + \sum_{k\neq i}B'_k\frac{-\{X_{kj}\}}{X_{ij} - \nint{X_{ij}}}.
\end{align}
This shows $B'X'_{*j} = C'_j = B_i$. For columns $\ell\neq j$ we get that 
\begin{align*}
    C'_\ell = C_\ell 
    &= BX_{*\ell} \\
    &= \sum_{k=1}^n B_k X_{k\ell} \\
    &= \sum_{k\neq i}B_k X_{k\ell} + B_i X_{i\ell} \\
    &\stackrel{(\ref{eq:old_Bi_in_terms_of_new_basis})}{=} \sum_{k\neq i}B'_k X_{k\ell} + \left( B'_i\frac{1}{X_{ij} - \nint{X_{ij}}} + \sum_{k\neq i}B'_k\frac{-\{X_{kj}\}}{X_{ij} - \nint{X_{ij}}} \right) \cdot X_{i\ell} \\
    &= \sum_{k\neq i}B'_k\cdot (X_{k\ell} - \frac{X_{i\ell}\{X_{kj}\}}{X_{ij} - \nint{X_{ij}}}) + B'_i\frac{X_{i\ell}}{X_{ij} - \nint{X_{ij}}} \\
    &= B'X'_{*\ell}.
\end{align*}
\end{proof}

For our target running time, we require a second adjustment. 
The exchange operation for updating $B_i$ after an exchange step requires $O(n^2)$ arithmetic operations. In order to reduce the number of arithmetic operations in \autoref{alg:solutionupdateversion}, we will delay updating the basis. Instead we will collect the representation of all exchange steps in a matrix $Y$, which is multiplied to the initial basis before output. 

\begin{algorithm}
  \caption{Generalized Euclidean Algorithm (Solution Updates)
    \label{alg:solutionupdateversion}}
  \begin{algorithmic}[1]
    \Statex{
    \textsc{Input: } A matrix $A=\left(A_1,\ldots, A_m \right) \in \zz^{n\times m}$ }
    \State{\textbf{find} independent vectors $B := \left(B_1, \ldots, B_n\right)$ with $B_i\in \{A_1,\ldots, A_m\}$}
    \State{\textbf{let} $C$ be a matrix with columns $\{A_1,\ldots, A_m\}\setminus \{B_1,\ldots, B_n\}$}
    \State{\textbf{compute} the inverse $B^{-1}$}
    \State{\textbf{compute} $X\leftarrow B^{-1}C$}
    \Let{$Y$}{$I_n$} \Comment{Invariant $B^{(\ell +1)} = B^{(1)}Y^{(\ell)}$}
    \While{$X$ not integral}
      \State{\textbf{choose minimal} $i\leq n$ and \textbf{any} $j\leq m-n$ s.t. $X_{ij}$ not integral}
      \State{$v_i \leftarrow X_{ij} - \nint{X_{ij}}$ and $v_k \leftarrow \{X_{kj}\}$ for all $k\neq i$}
      \Let{$Y$}{$Y \cdot \left(e_1, \ldots, e_{i-1}, v, e_{i+1}, \ldots, e_n \right)$} 
      \State{\textbf{update} X}
    \EndWhile
    \State \Return{$BY$}
  \end{algorithmic}
\end{algorithm}

\begin{theorem}\label{theo:solutionupdateversion}
\autoref{alg:solutionupdateversion} computes a basis for the lattice $\mathcal{L}(A)$ using \[\OTilde(\log\det(B^{(1)}) \cdot (m-n)n + mn^{\omega -1}\log\norm{A}_{\infty}) \]
arithmetic operations for an initial linearly independent subsystem $B^{(1)}$ found in line 1.
With the worst-case Hadamard bound on the determinant, the arithmetic complexity is
\[\OTilde((m-n)n^2\log\norm{A}_{\infty} + mn^{\omega -1}\log\norm{A}_{\infty}).\]
\end{theorem}
\begin{proof}
In order to prove correctness of the algorithm it suffices to show that the invariant $B^{(\ell+1)} = B^{(1)}Y^{(\ell)}$ holds, where $B^{(\ell)}, Y^{(\ell)}$, and $X^{(\ell)}$ represent the matrices $B$, $Y$, and $X$ starting iteration $\ell$, respectively. The exchange step in one iteration is $B^{(\ell+1)}_i = \sum_{k\neq i}B^{(\ell)}_k\{X^{(\ell)}_{kj}\} +  B^{(\ell)}_i (X^{(\ell)}_{ij}-\nint{X^{(\ell)}_{ij}})$ or in terms of the entire matrix it is
\begin{align*}
    B^{(\ell+1)} 
    &= B^{(\ell)} \cdot \left( e_1, \ldots, e_{i-1}, v^{(\ell)}, e_{i+1}, \ldots, e_n  \right) \\
    &= B^{(1)} \cdot \left( e_1, \ldots, e_{i-1}, v^{(1)}, e_{i+1}, \ldots, e_n  \right) \cdot \ldots \cdot \left( e_1, \ldots, e_{i-1}, v^{(\ell)}, e_{i+1}, \ldots, e_n  \right)\\
    &= B^{(1)}Y^{(\ell)}.
\end{align*}

Now, correctness of \autoref{alg:solutionupdateversion} follows similar to the proof of \autoref{theorem:basic_version_correctness} since $B$ and $C$ are updated just as in \autoref{alg:easyversion} and instead of computing a new solution in each iteration the complete solution matrix is updated in each iteration using \autoref{lem:updated_solution_matrix}. 

We find the set of linearly independent columns in time $\OTilde(mn^{\omega-1}\log\norm{A}_{\infty})$ using \autoref{lem:maximal_subsystem}.
The inverse and the matrix multiplication in lines 3 and 4 are computed in $\OTilde(n^\omega)$ and $ \OTilde(mn^{\omega -1})$, respectively. The number of iterations is bounded by $\log\det(B^{(1)})$ since in every iteration an exchange step is computed. Computing the vector $v$ requires $O(n)$ arithmetic operations. For $i' < i$ we have that $v_{i'} = \{X_{i'j}\} = 0$ since $i$ was chosen minimal considering fractional components of $X$. Computing $Y' \leftarrow Y \cdot \left(e_1, \ldots, e_{i-1}, v, e_{i+1}, \ldots, e_n \right)$ requires to compute 
\[Y'_i = Yv = \sum_{k \geq i}Y_k v_k.\]
A direct consequence of \autoref{lem:updated_solution_matrix} is that any integral row of the solution matrix $X_{kj}$ remains integral after the exchange step. Thus $Y_k = e_k$ for any $k>i$ and the computation simplifies to $Y'_i = Y_i v_i + \sum_{k \geq i}e_k v_k$ which can be computed with $O(n)$ arithmetic operations.  
In each iteration the main complexity is to update the $(m-n)n$ entries of $X$. Finally, in line 11 another matrix multiplication is performed in $\OTilde(n^\omega)$ arithmetic operations. The total running time is 
\[\OTilde(\log\det(B^{(1)}) \cdot (m-n)n + mn^{\omega -1}\log\norm{A}_{\infty}) \leq \OTilde((m-n)n^2\log\norm{A}_{\infty} + mn^{\omega -1}\log\norm{A}_{\infty}).\]
\end{proof}

\section{Bit Complexity}

A typical obstacle for computing the basis of a lattice is intermediate coefficient growth. Earlier algorithms for the HNF, for example, had their main computational bit complexity coming from intermediate numbers of length $\OTilde(n^4\log\norm{A}_\infty)$~\cite{DBLP:journals/siamcomp/KannanB79}. Later, all numbers involved could be bounded by $\det B \leq (n\norm{A}_\infty)^n$ for some subsystem $B$ of $A$, which still adds a factor of $n$. 


Large intermediate numbers could effect the bit complexity of our algorithm in two aspects: growing coefficients in the computed basis and exact solutions to linear systems. A naive implementation of our algorithmic idea could result in a basis with entries of exponential size. In every iteration, the new basis vector could be as large as the sum of the current basis vectors
\[\norm{B'_i}_\infty = \norm{\sum_{j\neq i}B_j\{x_j\} - B_i(x_i -\nint{x_i})}_\infty \leq  \sum_{j\leq n}\norm{B_j}_\infty \leq n\norm{B}_\infty.\]
If the initial basis is $B^{(1)}$, then there are up to $\log(\det B^{(1)})$ exchange steps. By Hadamard's bound coefficients in the basis might grow to be of order $(n\norm{A}_\infty)^{n}$ in a naive implementation. 

Fortunately, there is an easy pivoting rule that bounds the size of the computed basis $B$ by $\norm{B}_\infty \leq \OTilde(n^2\norm{A}_\infty)$. 
Our pivoting rule is very simple and in fact \autoref{alg:solutionupdateversion} already applies it. Instead of choosing any vector $c\in C$ and any fractional component of $x := B^{-1}c$, we compute exchange steps to obtain integral entries in the solution matrix $X:= B^{-1}C$ \emph{row by row}. If a row of the solution matrix is integral, then as a consequence of \autoref{lem:updated_solution_matrix} it remains integral after an exchange step. Moreover, in the modulo operation, basis vectors with integral solution component do not contribute to the new basis vector. If we assume that rows $i'<i$ of the solution matrix are integral we get that
\begin{align}\label{eq:bounded_basis_modulo_step}
    B'_i = \sum_{j\neq i}B_j\{x_j\} - B_i(x_i -\nint{x_i}) = \sum_{j> i}B_j\{x_j\} - B_i(x_i -\nint{x_i}).
\end{align}
Performing modulo and exchange steps row by row in the solution matrix corresponds to column by column in the current basis. Therefore, the basis vectors $B_j$ with $j<i$ are final in the sense that those will appear in the output basis and the basis vectors $B_j$ with $j>i$ are untouched in the sense that they were part of the input vectors which implies that their size is bounded by $\norm{B_j}_\infty \leq \norm{A}_\infty$. By Equation~\ref{eq:bounded_basis_modulo_step} only the  untouched basis vectors with the before mentioned size bound and the currently updated basis vector contribute to the new basis vector. Therefore, the size of the modulo vector is bounded by 
\[\norm{B'_i}_\infty \leq  (n-1)\norm{A}_\infty + \norm{B_i} \leq n^2\norm{A}_\infty\log(n\norm{A}_\infty)\] 
since there are at most $\log(\det B^{(1)}) \leq n\log(n\norm{A}_\infty)$ exchange steps.

Using this pivoting rule, large numbers may only appear as a result of exact system solving.  By Cramer's rule and Hadamard's bound exact solutions to a linear system $Bx = b$ can be as large as $n^{n/2}\norm{B}_{\infty}^{n-1}\norm{b}_\infty$ in the numerator and $n^{n/2}\norm{B}_{\infty}^n$ in the denominator. We use the recent algorithm by Birmpilis, Labahn and Storjohann to compute solutions of linear systems.

\begin{theorem}[\cite{DBLP:conf/issac/BirmpilisLS19,DBLP:conf/issac/LiS22}]\label{theo:system_solving}
There exists an algorithm that takes as input a non-singular matrix $B\in \zz^{\Tilde{n}\times \Tilde{n}}$ and a vector $b\in\zz^{\Tilde{n}}$ and returns as output $B^{-1}b \in \qq^{\Tilde{n}}$. If $\log\norm{b}_\infty \in \OTilde(\Tilde{n}\log\norm{B}_\infty)$, the running time of the algorithm is $\OTilde(\Tilde{n}^\omega\norm{B}_\infty)$ bit operations.
\end{theorem}

We use the following lemma for calculations involving a vector with large coefficients such as computation of the remainder of our modulo operations.

\begin{lemma}[\cite{DBLP:conf/issac/BirmpilisLS19}]\label{lem:dimension_size_tradeoff}
Let $B\in\zz^{\Tilde{n}\times \Tilde{n}}$ and $N\in\zz_{>0}$ be a power of $2$ such that $\log N \in O(\log(\Tilde{n}\norm{B}_\infty))$. If $C\in\zz/(N^p)^{\Tilde{n}\times \Tilde{m}}$ with $\Tilde{m} p\in O(\Tilde{n})$, then $\rem(BC, N^p)$ can be computed in with bit complexity \[\OTilde(\Tilde{n}^\omega\log\norm{B}_{\infty}).\]
\end{lemma}

In order to quickly perform our pivoting rule, a new subproblem arises. We need to locate the next modulo and exchange step and thus require to efficiently find non-integral components of a row of the solution matrix $X := B^{-1}C$. The following lemma shows that a row of the solution matrix can be computed with similar bit complexity as a column.

\begin{lemma}\label{lem:solve_row}
Consider a full rank matrix $B\in\zz^{\Tilde{n}\times \Tilde{n}}$,  a matrix $C\in\zz^{\Tilde{n}\times \Tilde{m}}$, and $\delta \in \nn$ such that $\norm{B}_\infty \leq \delta$ and $\norm{C}_\infty \leq \delta$. Let $X\in \qq^{\Tilde{n}\times\Tilde{m}}$ be the solution matrix for $BX=C$. Any row $i\leq \Tilde{n}$ of the solution matrix $X$ can be computed using $\OTilde(\max\{\Tilde{n},\Tilde{m}\}\Tilde{n}^{\omega -1}\log\delta)$ bit operations.
\end{lemma}
\begin{proof}
The procedure is as follows. First, we compute $y\in\qq^{\Tilde{n}}$ such that $B^\intercal y = e_i$. This is the same as the $i$th row of $Y := I_{\Tilde{n}} B^{-1}$, where $I_{\Tilde{n}}$ is the identity matrix of dimension ${\Tilde{n}}$. In other words, $Y$ is the inverse of $B$ and $y$ is the $i$th row of the inverse of $B$. Then we compute an integer $\mu \leq \det B$ such that $\mu y$ is integral.  Finally, we compute $z$ such that $\frac{1}{\mu}C^\intercal (\mu y)$. It is obviously the same to compute $\frac{1}{\mu}\mu y^\intercal C = y^\intercal C $. Since $y$ is the $i$th row of the inverse of $B$ we have that $z$ is the $i$th row of the solution matrix $X=B^{-1}C$.

We can compute $y$ with \autoref{theo:system_solving} using $\OTilde(\Tilde{n}^\omega\log\delta)$ bit operations. The integer $\mu$ can be found in $\OTilde(\Tilde{n}^2\log\delta)$ bit operations.\footnote{One way to do this is as follows. Let $d_1,\ldots,d_{\Tilde{n}}$ be the denominators of $y$. Compute the greatest common divisor of $d_1$ and $d_2$ and $\frac{d_1\cdot d_2}{\mathrm{gcd}(d_1,d_2)}$ to obtain the least common multiple of $d_1$ and $d_2$. Continue with the least common multiple and $d_3$ and eventually obtain the least common multiple of $d_1, \ldots,d_{\Tilde{n}}$. Due to Cramer's rule $\mu := \mathrm{lcm}(d_1, \ldots,d_{\Tilde{n}})$ is at most $\det(B)$.}
The matrix vector multiplication to compute $z$ can be done in $\OTilde(\max\{\Tilde{n},\Tilde{m}\}\Tilde{n}^{\omega-1}\log\delta)$ bit operations using \autoref{lem:dimension_size_tradeoff} $\lceil\frac{\Tilde{n}}{\Tilde{m}}\rceil$ times for $p:= \Tilde{n}$. Scaling $y$ and the result of the matrix-vector multiplication each costs $\OTilde(\Tilde{n}^2\log\delta)$ bit operations. 
\end{proof}

\subsection{General Version}

\begin{algorithm}
  \caption{Generalized Euclidean Algorithm (Bit Complexity)
    \label{alg:easyversion_small_numbers}}
  \begin{algorithmic}[1]
    \Statex
    \textsc{Input: } A matrix $A=\left(A_1,\ldots, A_m \right) \in \zz^{n\times m}$ \\
    \textbf{find} independent vectors $B := \left(B_1, \ldots, B_n\right)$ with $B_i\in \{A_1,\ldots, A_m\}$
    \Let{$C$}{$\{A_1,\ldots, A_m\}\setminus \{B_1,\ldots, B_n\}$}
    \Let{$i$}{$1$}
    \While{$i\leq n$}
        \State{\textbf{compute} the $i$th row $z$ of the solution matrix $X:= B^{-1}C$}
        \If{$z$ integral}
            \Let{$i$}{$i+1$}
        \Else
            \State{\textbf{choose any} $j$ such that $z_j$ is not integral}
            \State{\textbf{solve} $Bx =C_j$}
            \State{$C\leftarrow C\setminus\{C_j\}\cup \{B_i\}$}
            \Let{$B_i$}{$C_j - (\sum_{j\neq i}B_j\floor{x_j}  + B_i \nint{x_i})$}
        \EndIf
        \EndWhile
    \State \Return{$B$}
  \end{algorithmic}
\end{algorithm}

\begin{theorem}\label{theorem:easy_version_complexity}
\autoref{alg:easyversion_small_numbers} computes a basis for the lattice $\mathcal{L}(A)$ using at most $\OTilde(mn^\omega\log^2\norm{A}_{\infty})$ bit operations.
\end{theorem} 
\begin{proof}

We want to prove correctness by proving that \autoref{alg:easyversion_small_numbers} performs the exchange steps from \autoref{alg:easyversion} but in a more specified order.

\begin{addmargin}[2em]{0em}
\noindent\textit{Claim.} Consider iteration $i\leq n$. For any $i'<i$ all solutions $x^{(c)} = B^{-1}c$ for $c\in C$ are integral at index $i'$.
\end{addmargin}

\noindent By lines 6 and 7 solution index $x^{(c)}_i$ for $Bx^{(c)}=c$ is integral for all $c\in C$ when $i$ is set to $i+1$ in line 7. Thus, we need to prove that this remains true after an exchange step in lines 11-12. 

Consider any $i'<i$ and right hand side $c$ used for the exchange step. Let $x^{(c)} := B^{-1}c$ and $x^{(c')}:= B^{-1}c'$ for any $c' \in C$ with $c'\neq c$. By \autoref{lem:updated_solution_matrix} the updated solutions of $c$ and $c'$ at index $i$ are 
\begin{align}
    \frac{-\{x^{(c)}_{i'}\} }{x^{(c)}_i - \nint{x^{(c)}_i}}  \quad\text{and}\quad  x^{(c')}_{i'} + \frac{-x^{(c')}_i \cdot \{x^{(c)}_{i'}\} }{x^{(c)}_i - \nint{x^{(c)}_i}},\quad\text{respectively.} 
\end{align}
Since after iteration $i$ all solutions are integral at index $i$ by lines 6-7, this implies that the exchange step in line 11-12 does keep the property that for any $i'<i$ all solutions are integral at index $i'$. \medskip

The claim implies that all solutions are integral when the algorithm terminates. Therefore, correctness of the algorithm follows from \autoref{theorem:basic_version_correctness} since \autoref{alg:easyversion_small_numbers} only selects the next exchange step in a specified order compared to \autoref{alg:easyversion}.

Concerning the running time, we start off by bounding size of the numbers involved.  By the definition of the exchange step any new vector 
\[B_i = \sum_{j\neq i}B_j\{x_j\} + B_i(x_i - \nint{x_i})\] 
is the sum of $B_j y_j$ for $|y_j|\leq 1$ and $j\geq i$ since for any $j<i$ the solution at index $j$ is integral by the claim and thus the fractional component is $0$. Any $B_j$ for $j > i$ is unchanged after line 1 and thus $\norm{B_j}_\infty \leq \norm{A}_{\infty}$. If $B_i$ and $B_i'$ are the state of the $i$th vector before and after the exchange step in line 11-12, respectively, then 
\begin{align}
    \norm{B_i'}_\infty \leq \norm{B_i}_\infty + \sum_{j>i}\norm{B_j}_{\infty} \leq \norm{B_i}_\infty + (n-1)\norm{A}_{\infty}.
\end{align} 
Let $B^{(1)}$ be the basis in line 1 and $B^{(\ell)}$ the returned basis. Since there are in total at most $\log\det(B^{(1)})$ exchange steps, the returned (and every intermediate) basis is bounded by
\begin{align*}
    \norm{B^{(\ell)}}_\infty \leq \log\det(B^{(1)}) \cdot n\norm{A}_{\infty} \leq n^2\norm{A}_{\infty} \log(n\norm{A}_{\infty}) =: \delta.
\end{align*}

By Hadamard's inequality and Cramer's rule the numerator and denominator of solutions $x = B^{-1}c$ for $c\in C$ are bounded by determinants of $B$ and $B|c$, where one column of $B$ is exchanged by $c$, respectively, and due to the bounded entry size this is $\leq (n\delta)^n$ in every iteration. 

By \autoref{lem:maximal_subsystem} the set of linearly independent columns can be found with $\OTilde(mn^{\omega-1}\log\norm{A}_{\infty})$ bit operations.  
Every iteration of the while loop either increases $i$ or performs an exchange step. Hence, there are at most $n + \log \det (B^{(1)}) = O(n\log(n\norm{A}_{\infty}))$ iterations. 

The $i$th row of the solution matrix can be found in $\OTilde(\max\{m-n,n\}n^{\omega-1}\log\delta)$ bit operations using \autoref{lem:solve_row}. In line 10 a linear system is solved. All numbers involved are bounded by $\delta$ and thus the linear system can be solved in $\OTilde(n^\omega\log\delta)$ bit operations. Considering line 12, let $\Tilde{x} \in \zz^n$ be defined as $\Tilde{x}_j = \floor{x_j}$ for $j\neq i$ and $\Tilde{x_i} := \nint{x_i}$. The updated column is then $B_i \leftarrow C_j - B\Tilde{x}$, where the latter can be computed in $\OTilde(n^\omega \log\delta)$ bit operations using \autoref{lem:dimension_size_tradeoff} since $\norm{\Tilde{x}}_\infty \leq (n\delta)^n$ and can be scaled to an integral vector similar to the proof of \autoref{lem:solve_row}. 
Overall the number of bit operations for \autoref{alg:easyversion_small_numbers} is bounded by 
\begin{align*}
    \OTilde((n+\log(\det B^{(1)}))\cdot((m-n)n^{\omega-1}\log\delta + n^{\omega}\log{\delta}))&=\OTilde((n\log\delta)mn^{\omega-1}\log\delta )\\ &= \OTilde(mn^{\omega}\log^2\norm{A}_{\infty}).
\end{align*}

\end{proof}

\subsection{Few Additional Vectors $m-n$}
Very recently Lin and Storjohann considered the special case that $m-n=k$ for a constant $k$~\cite{DBLP:conf/issac/LiS22}. In this section we present a variant of our generalized Euclidean algorithm that improves the general running time in the case that $m-n$ is small but not necessarily constant, e.g. the running time dependence on $m$ and $n$ is improved for any instance with $m-n \in \OTilde(n^{\omega -2})$. The procedure is almost identical to \autoref{alg:solutionupdateversion}.

\begin{algorithm}
  \caption{Generalized Euclidean Algorithm (Bit Complexity, $m-n$ small)
    \label{alg:easyversion_small_numbers_m-n-small}}
  \begin{algorithmic}[1]
    \Statex{
    \textsc{Input: }  A matrix $A=\left(A_1,\ldots, A_m \right) \in \zz^{n\times m}$ }
    \State{\textbf{find} independent vectors $B := \left(B_1, \ldots, B_n\right)$ with $B_i\in \{A_1,\ldots, A_m\}$}
    \State{\textbf{let} $C$ be a matrix with columns $\{A_1,\ldots, A_m\}\setminus \{B_1,\ldots, B_n\}$}
    \State{\textbf{compute} $X\leftarrow B^{-1}C$}
    \Let{$Y$}{$I_n$} \Comment{Invariant $B^{(\ell +1)} = B^{(1)}Y^{(\ell)}$}
    \While{$X$ not integral}
      \State{\textbf{choose minimal} $i\leq n$ and \textbf{any} $j\leq m-n$ s.t. $X_{ij}$ not integral}
      \State{$v_i \leftarrow X_{ij} - \nint{X_{ij}}$ and $v_j \leftarrow \{X_{kj}\}$ for all $k\neq i$}
      \Let{$Y$}{$Y \cdot \left(e_1, \ldots, e_{i-1}, v, e_{i+1}, \ldots, e_n \right)$} 
      \State{\textbf{update} X}
    \EndWhile
    \State \Return{$BY$}
    \end{algorithmic}
\end{algorithm}

\begin{theorem}\label{theo:small_m-n}
\autoref{alg:easyversion_small_numbers_m-n-small} computes the basis of the lattice $\lattice(A)$ using \[\OTilde((m-n)n^3\log^2\norm{A}_{\infty} + n^{\omega(2)}\log\norm{A}_{\infty})\] 
bit operations.

\end{theorem}

\begin{proof}
The size of most intermediate numbers is bounded as in \autoref{theorem:easy_version_complexity}. Additionally, we need to bound the size of numbers in $Y$. Rephrased, $Y$ is the solution matrix for $B^{(1)}Y= B^{(\ell)}$. The size of numbers in $B^{(\ell)}$ is bounded by $\delta = n^2\norm{A}_{\infty}\log(n\norm{A}_{\infty})$ as in the proof of \autoref{theorem:easy_version_complexity}. Thus, denominators in $Y$ are bounded by $\det B^{(1)}$ and numerators in $Y$ are bounded by $\det (B^{(1)}|B^{(\ell)}_k) \leq (n\norm{A}_{\infty})^{O(n)}$, where the latter describes the matrix exchanging a column of $B^{(1)}$ with $B^{(\ell)}_k$. 

The set of independent vectors can be computed in the claimed time using \autoref{lem:maximal_subsystem}. For line 3 we compute $m-n$ solutions to linear systems. This is also the claimed time by \autoref{theo:system_solving}. Updating $Y$ costs $O(n)$ arithmetic operations as analyzed in \autoref{theo:solutionupdateversion}. Let $\delta$ be as in the proof of \autoref{theorem:easy_version_complexity}. The size of numbers involved is bounded by $\OTilde(n\log\delta)$ and thus lines 8-9 require $\OTilde(n^2\log\norm{A}_{\infty})$ bit operations. Updating the solution matrix $X$ can be done with \autoref{lem:updated_solution_matrix} using $O((m-n)n)$ arithmetic operations and due to the bounded size of numbers this requires at most $\OTilde((m-n)n^2\log\norm{A}_{\infty})$ bit operations. In every iteration of the while loop an exchange operation is performed. 
Thus, there are at most $\log(\det B^{(1)})$ iterations. 

In line 10 we can multiply the matrix by the least common multiple of the denominators (which is bounded by $\det B^{(1)}$), apply the matrix multiplication, and again divide by the least common multiple of denominators, similar to part of the proof of \autoref{lem:solve_row}.
Then the matrix multiplication can be solved by idea of \autoref{lem:dimension_size_tradeoff} the main complexity is to compute a matrix multiplication of dimensions $n\times n$ and $n\times n\cdot O(n)$. Using rectangular matrix multiplication~\cite{DBLP:conf/soda/GallU18,DBLP:conf/soda/AlmanW21} this can be done using $\OTilde(n^{\omega(2)}\log\norm{A}_{\infty})$ bit operations, where $\omega(k)$ is the exponent required to compute a matrix multiplication for dimensions $n\times n$ and $n\times n^k$. 


The bit complexity in total is bounded by
\begin{align*}
    \OTilde((m-n)n^{\omega}\log\norm{A}_{\infty} + mn^{\omega -1}\log\norm{A}_{\infty} + (m-n)n^2\log\det(B^{(1)})\log\norm{A}_{\infty} + n^{\omega(2)}\norm{A}_{\infty}).
\end{align*}
Using $(m-n)\in O(n^{\omega - 2})$ and the worst-case Hadamard bound on the determinant the running time simplifies to $\OTilde((m-n)n^3\log^2\norm{A}_{\infty} + n^{\omega(2)}\norm{A}_{\infty})$ bit operations.
\end{proof}

\subsection{Small $n\times n$ Minors}

In this section we give an algorithm which is very efficient in the case that $\det B^{(1)}$ is small. This is often the case when considering specific matrix classes. For example, a prominent class of matrices that is often considered in integer programming, is the class of matrices $A$ where the absolute value of all subdeterminants are bounded by some small $\Delta$.

Now consider again \autoref{alg:easyversion_small_numbers_m-n-small}. The number of iterations of the while loop scales the complexity by $\log(\det B^{(1)})$. So, if the determinant is small, the bit complexity for lines 5-10 also decreases. In contrast to other algorithms~\cite{DBLP:journals/siamcomp/HafnerM91, DBLP:conf/issac/StorjohannL96} the following algorithm directly benefits from small minors and does not require the approximate size as input or to compute any determinant. In order to achieve an improved running time, we analyze the algorithm for solving a system of linear equations from \cite{DBLP:conf/issac/BirmpilisLS19} for a matrix right-hand side.

\begin{corollary}\label{theo:solution_matrix}
The algorithm \texttt{solve} in \cite{DBLP:conf/issac/BirmpilisLS19} solves a system $X = B^{-1}C$ for an invertible matrix $B\in\zz^{{\Tilde{n}} \times {\Tilde{n}}}$ and a matrix $C \in\zz^{{\Tilde{n}} \times \Tilde{m}}$ using $\OTilde(\frac{{\Tilde{m}}\log(\Delta)}{{\Tilde{n}}}{\Tilde{n}}^{\omega}\log\delta) = \OTilde({\Tilde{m}}\log(\Delta){\Tilde{n}}^{\omega-1}\log\delta)$ bit operations, where $\Delta$ is the largest ${\Tilde{n}}\times {\Tilde{n}}$ minor of $(B_1,\ldots, B_{\Tilde{n}}, C_1,\ldots, C_{\Tilde{m}})$ and $\norm{B}_{\infty}\leq \delta$ and $\norm{C}_{\infty}\leq \delta$.
\end{corollary}
\begin{proof}
We analyze their algorithm and how the running time changes by the modification in their notation. Also we only describe the differences in the analysis. On a high level, the main change is that we do not provide the so called dimension $\times$ precision invariant but instead parameterize by this quantity. We throughoutly make use of the dimension $\times$ precision tradeoff, where the idea of \autoref{lem:dimension_size_tradeoff} is used for matrix multiplications and since the size of numbers is bounded by $\Delta$ this results in matrix multiplications of dimension $\Tilde{n} \times \Tilde{n}$ and $\Tilde{n} \times \Tilde{m} \log\Delta$ with sufficiently bounded coefficients. Viewed as $\frac{\Tilde{m} \log\Delta}{\Tilde{n}}$ matrix multiplications the running time follows. 

If we analyse the algorithm for a matrix right-hand side, steps 1 and 2 do not change. In step 3 the subroutine \texttt{SpecialSolve} dominates the running time.  Corollary 7 in their paper requires the dimension $\times$ precision invariant $\Tilde{m}\cdot \log(\Delta) \in O(\Tilde{n}\log(\Tilde{n}\delta))$, which is not necessarily the case here. However, the running time is dominated by $\OTilde(\log\log(\Delta))$ matrix multiplications of an $\Tilde{n}\times \Tilde{n}$ matrix  with coefficients of magnitude $O(\Tilde{n}^2\norm{B})$ and an $\Tilde{n}\times \Tilde{m}$ matrix of magnitude $\Delta$ as by Cramer's rule numbers involved in this step are bounded by $\det B$ and $\det (B|C_i)$. Using \autoref{lem:dimension_size_tradeoff} (their Lemma~2) this can be computed in target time. 


Finally, Step 4 consists of matrix multiplications $\rem(PM(2^eS^{-1})Y, 2^d)$, where $d\in \OTilde(\Tilde{n}\log\delta)$. The first part $Z := (2^eS^{-1})Y$ involves just a diagonal matrix $S^{-1}$ and can be computed in time. For the multiplication $MZ$, we follow the steps from their paper. By their Lemma 17, the X-adic expansion of the columns of $M$ consists of $\Tilde{n}' \leq 2\Tilde{n}$ columns for $X$ the smallest power of $2$ such that $X\geq \sqrt{\Tilde{n}}\norm{A}$. Let $M' = \left(M_0 \ldots M_{p-1} \right)$ be the $X$-adic expansion of $M$, where $M_i\in\zz^{\Tilde{n}\times k_i}$ and $\sum{i <p}k_i = \Tilde{n}' \leq 2\Tilde{n}$. Let $ Z = \left( Z_0 \ldots Z_{p-1} \right)$ be the X-adic expansions of $Z$ and let $Z_i^{(k_i)}$ be the submatrix of the last $k$ rows. The matrix multiplication can be restored from the product
\begin{align}
    \begin{pmatrix}M_0 \ldots M_{p-1}\end{pmatrix}
    \begin{pmatrix}
    Z_0^{(k_0)} & Z_1^{(k_0)} & \ldots & Z_{p-1}^{(k_0)} \\
    & Z_0^{(k_1)} & \ldots & Z_{p-2}^{(k_1)} \\
    && \ddots & \vdots \\
    &&& Z_0^{(k_{p-1})}
    \end{pmatrix}.
\end{align}
The dimensions are $\Tilde{n} \times \Tilde{n}'$ and $\Tilde{n}' \times \Tilde{m}\log(\Delta)$ since the precision $p$ requires $p\geq \log \norm{Z}_\infty$ which is bounded by Cramer's rule. 

\end{proof}

Though, we already analyzed \autoref{alg:easyversion_small_numbers_m-n-small} for small $m-n$, next we will analyze it again for the case that all $n\times n$ minors of the input are small.

\begin{theorem}
\autoref{alg:easyversion_small_numbers_m-n-small} computes a basis for the lattice $\mathcal{L}(A)$ and the running time is \[\OTilde((m-n)n^2\log\delta\log\Delta + n\log^3\Delta + n^{\omega}\log^2\Delta)\] bit operations for $\Delta$ being the largest $n\times n$ minor of $A$.
\end{theorem}
\begin{proof}
The size of most intermediate numbers is bounded as in \autoref{theorem:easy_version_complexity}. Additionally, we need to bound the size of numbers in $Y$ and the bound from \autoref{theo:small_m-n} does not suffice. Let $Y^{(\ell)}$ be the state of $Y$ in the $i$th iteration. Consider in iteration $\ell$ where $X^{(\ell)}_{ij}$ was chosen for the modulo operation. Updating $Y^{\ell}$ only changes $Y^{\ell+1}_i = Y^{\ell}v$. The update for row index $k \leq i$ is 
\begin{align*}
    Y^{\ell +1}_{ki} &= \sum_{h \leq n}Y^{\ell}_{kh}v_h\\
    &= \sum_{h \neq i}Y^{\ell}_{kh}\{X_{hj}\} + Y^{\ell}_{ki}(X^{\ell}_{ij} - \nint{X^{\ell}_{ij}}) \\
    &= Y^{\ell}_{ki}(X^{\ell}_{ij} - \nint{X^{\ell}_{ij}})\\
    &= \frac{\det B^{(\ell+1)}}{\det B^{(\ell')}},
\end{align*}
where $\ell'$ is the first iteration considering row $i$. The update for row indices $k >i$ is 
\begin{align*}
    Y^{\ell +1}_{ki} &= \sum_{h \leq n}Y^{\ell}_{kh}v_h\\
    &= \sum_{h \neq i}Y^{\ell}_{kh}\{X_{hj}\} + Y^{\ell}_{ki}(X^{\ell}_{ij} - \nint{X^{\ell}_{ij}}) \\
    &= 1 \cdot \{X_{kj}\} + Y^{\ell}_{ki}\frac{\det B^{(\ell +1)}}{\det B^{\ell}}. 
\end{align*}
The denominators of $\{X_{kj}\}$ and $\frac{\det B^{(\ell +1)}}{\det B^{\ell}}$ are both divisors of $\det B^{\ell}$. Therefore numerator and denominator of $Y^{\ell +1}_{ki}$ are both bounded by $\Delta^{O(\log\det (B^{1}))}$.

The set of linearly independent columns can be found using $\OTilde(mn^{\omega-1}\log\norm{A}_{\infty})$ bit operations with \autoref{lem:maximal_subsystem}. The solution matrix $X$ can be computed using \autoref{theo:solution_matrix} in $\OTilde((m-n)\log(\det B)n^{\omega-1}\log\norm{A}_{\infty})$ bit operations. There are at most $\log(\det B)$ iterations of the while loop and all intermediate bases are bounded by $\delta$ in infinity norm. Thus, computing $v$ and updating $Y$ requires $O(n)$ arithmetic operations, which are at most $\OTilde(n^2\log\norm{A}_{\infty} + n\log^2\Delta)$ bit operations, depending on whether the bound on $X$ or the bound on $Y$ is larger. Updating $X$ requires $O((m-n)n)$ arithmetic operations and thus $\OTilde((m-n)n^2\log\norm{A}_{\infty})$ bit operations. Finally, line 10 computes the basis using a matrix multiplication. By the bound on the size of numerators and denominators of $Y$ we have that both are at most $O(\log^2\Delta)$, which also applies to the least common multiple of denominators. Thus, the matrix multiplication can be computed using $\OTilde(n^{\omega}\log^2(\Delta))$ bit operations. 



The total running time in bit operations is therefore bounded by 
\begin{align*}
    \OTilde((m-n)n^2\log\norm{A}_{\infty}\log\Delta + n\log^3\Delta + n^{\omega}\log^2\Delta).
\end{align*}
In the case that $\log\Delta \in O(n^{\omega -2})$ the running time simplifies to $\OTilde((m-n)n^2\norm{A}_{\infty}\log\Delta + n^{\omega}\log^2\Delta)$.
\end{proof}

\section{Modifications of the Algorithm}
In this section, we present how our algorithms can be modified to compute the determinant of a square matrix $B$ or to compute a solution of Diophantine system of equations.

\subsection{Computing the Determinant}
Our algorithms can be easily adapted to compute the determinant of a matrix $B$ for a given matrix $B \in \zz^{n \times n}$. We initialize the respective algorithm with the matrix $A = (B I)$, where $I$ is the identity matrix. By this we ensure that $\det(\lattice(A)) = 1$. The first line of finding a set of linearly independent vectors is skipped and instead set to $B$. 
Then we can simply keep track of the improvement to $\det B$ after each exchange operation. As explained, by Cramer's rule the determinant of the new basis $B'$ equals $(x_i - \lceil x_i \rfloor) \cdot \det B$. Multiplying the improvements over all exchange operations therefore yields the value $\det(B)/ \det(\lattice(A))$. For this, we only have to introduce a new variable $D$ and set $D = D \cdot (x_i - \lceil x_i \rfloor)$ whenever there is an exchange operations.

The running time of the respective algorithm remains the same with $m, m-n \in O(n)$.

\subsection{Solving Systems of Diophantine Equations}
The problem of solving a system of Diophantine equations is to compute $x$ such that
\begin{align} \label{dioph_equ}
    Ax = b\\
    x \in \zz^m.
\end{align}
for a given matrix $A \in \zz^{n \times m}$ and vector $b \in \zz^n$.

The classical Euclidean algorithm can be extended to compute $x,y \in \zz$ such that \begin{align*}
    ax + by = \gcd(a,b)
\end{align*}
and therefore solve Diophantine Equations of the form $a_1 x_1 + \ldots a_n x_n = b$, by applying the algorithm iteratively.
Similarly, our algorithm can be extended to compute a basis matrix $B \in \zz^{n \times n}$ with $\lattice(A) = \lattice(B)$ and matrix $U \in \zz^{m \times n}$ such that
\begin{align*}
    A U = B.
\end{align*}
Using $U$, one can solve (\ref{dioph_equ}) by first solving the linear system of equations $Bx = b$. The Diophantine equation (\ref{dioph_equ}) is feasible if and only if $x$ is integral and a solution to (\ref{dioph_equ}) is then given by $U x$.

The computation of $U$ can be realised very similar to the computation of $Y$ in e.g. \autoref{alg:solutionupdateversion}. Initially set the columns of $U$ to $e_i$ if this column of $A$ is the $i$th column of the initial basis $B$. For an exchange step
\begin{align*}
    B_i = c - (\sum_{j\neq i}B_j\floor{x_j} + B_i\nint{x_i})
\end{align*}
set $v_{j'} = -\floor{x_j}$, $v_{i'} = -\nint{x_i}$ and $v_{k} = 1$ if $c$ is the $k$th column and $j'$ and $i'$ are the current indices for the columns of $j$ and $i$, respectively. The exchange step can be expressed in $U$ as $U_i = Uv$. The index for basis column $i$ changes to $k$. However, note that this procedure requires an additional term of $\OTilde(mn^2\log(\det B^{(1)})\log\norm{A}_\infty)$ in bit complexity.

\section{Conclusion and Future Research}
Our novel approach for lattice basis computation provides the first running time improvement since 1996 based on a generalization of the Euclidean algorithm. However, this improvement applies only if we count arithmetic operations. A natural direction for future research would be to investigate whether this approach can improve also on the bit complexity in general. A similar approach like Schönhage~\cite{DBLP:journals/acta/Schonhage71} for the classical Euclidean algorithm might also work for the generalization that we presented.

Furthermore, it would be interesting to see how the algorithms perform in practice. Given that the determinant of the initial basis matrix $B$ should be smaller than the worst case Hadamard bound in most practical instances, our algorithms might actually perform rather well. Moreover, the improvement on the determinant on average in practice will be much better than $1/2$.

\bibliographystyle{alpha}
\bibliography{references}

\begin{acronym}
\acro{hnf}[HNF]{Hermite normal form}
\acro{snf}[SNF]{Smith normal form}
\acro{de}[DE]{Diophantine Equations}
\acro{lbr}[LBR]{Lattice Basis computation}
\end{acronym}

\end{document}